\documentclass[letterpaper, 10 pt, conference]{ieeeconf}  

\IEEEoverridecommandlockouts                              
\usepackage{amsmath} 
\usepackage{amsfonts}

\newtheorem{proposition}{Proposition}[section]
\newtheorem{exmp}{Example}[section]
\usepackage{float} 
\usepackage[skip=1pt,font=footnotesize]{caption}

\usepackage{tikz}
\usetikzlibrary{shapes,arrows,calc,positioning}
\tikzstyle{bigblock} = [draw, fill=blue!20, rectangle, 
    minimum height=6em, minimum width=8em]
\tikzstyle{medblock} = [draw, fill=blue!20, rectangle, 
    minimum height=4em, minimum width=4em]    
\tikzstyle{mux} = [draw, fill=black!20, rectangle, 
    minimum height=5em, minimum width=0.1em]    
\tikzstyle{smallblock} = [draw, fill=blue!20, rectangle, 
    minimum height=2em, minimum width=3em]
    
\tikzstyle{data_block} = [draw, fill=green!20, rectangle, 
    minimum height=2em, minimum width=3em]
\tikzstyle{ops_block} = [draw, fill=blue!20, rectangle, 
    minimum height=2em, minimum width=3em]    
\tikzstyle{est_block} = [draw, fill=red!20, rectangle, 
    minimum height=2em, minimum width=3em]    
    
\tikzstyle{sum} = [draw, fill=blue!20, circle, node distance=1cm,minimum height=0.5cm]
\tikzstyle{signal} = [coordinate]
\tikzstyle{pinstyle} = [pin edge={to-,thin,black}]
\tikzstyle{block} = [draw, fill=blue!20, rectangle, 
    minimum height=3em, minimum width=9em]
\tikzstyle{blockS} = [draw, fill=blue!20, rectangle, 
    minimum height=3em, minimum width=4em]    
\tikzstyle{input} = [coordinate]
\tikzstyle{output} = [coordinate]
\usetikzlibrary{matrix}

\usetikzlibrary{positioning}
\usetikzlibrary{math}

\usepackage{hyperref}
\usepackage{xcolor}
\hypersetup{
    colorlinks,
    linkcolor={blue!100!black},
    citecolor={blue!50!black},
    urlcolor={blue!80!black}
}

\newcommand{\bc}{\begin{center}}
\newcommand{\ec}{\end{center}}
\newcommand{\benum}{\begin{enumerate}}
\newcommand{\eenum}{\end{enumerate}}
\newcommand{\nn}{\nonumber}
\newcommand{\matl}{\left[ \begin{array}}
\newcommand{\matr}{\end{array} \right]}
\newcommand{\matls}{\left[ \begin{smallmatrix}}
\newcommand{\matrs}{\end{smallmatrix} \right]}
\newcommand{\isdef}{\stackrel{\triangle}{=}}

\newcommand{\tr}{{\rm tr}\,}

\newcommand{\rmT}{{\rm T}}

\newcommand{\rmd}{{\rm d}}

\newcommand{\rms}{{\rm s}}

\newcommand{\BBE}{{\mathbb E}}

\newcommand{\BBR}{{\mathbb R}}

\newcommand{\SN}{{\mathcal N}}

\newcommand{\neweqline}{\ensuremath{\nn \\ &\quad }}

\usepackage[style=ieee ,sorting=none,maxbibnames=99,giveninits, doi=false]{biblatex}
\title{\LARGE \bf
Modified Unscented Kalman Filter
}

\title{\LARGE \bf
Two Modifications of the Unscented Kalman Filter that\\ Specialize to the Kalman Filter for Linear Systems
}

\author{Ankit Goel and Dennis S. Bernstein
\thanks{Ankit Goel is a postdoctoral research fellow in the Department of Aerospace Engineering, University of Michigan, Ann Arbor, MI 48109. {\tt\small ankgoel@umich.edu}}%
\thanks{Dennis Bernstein is a Professor in the Department of Aerospace Engineering, University of Michigan, Ann Arbor, MI 48109. {\tt\small dsbaero@umich.edu}}%
}

\begin{document}

\maketitle
\thispagestyle{empty}
\pagestyle{empty}

\begin{abstract}
%
Although the unscented Kalman filter (UKF) is applicable to nonlinear systems, it turns out that, for  linear systems, UKF does not specialize to the classical Kalman filter.
This situation suggests that it may be advantageous to modify UKF in such a way that, for linear systems,   the Kalman filter is recovered.
The ultimate goal is thus to develop modifications of UKF that specialize to the Kalman filter for linear systems and have improved accuracy for nonlinear systems.
With this motivation, this paper presents two modifications of UKF that specialize to the Kalman filter for linear systems.
The first modification (EUKF-A) requires the Jacobian of the dynamics map, whereas the second modification (EUKF-C) requires the Jacobian of the measurement map.   
For various nonlinear examples, the accuracy of EUKF-A and EUKF-C  is compared to the accuracy of UKF.

\end{abstract}


\section{INTRODUCTION}

The Unscented Kalman filter (UKF) is widely applied to  nonlinear estimation problems \cite{simon2006optimal}.
UKF was introduced in \cite{wan2000unscented,van2001square} has been applied to attitude estimation  \cite{kraft2003quaternion}, 
navigation  \cite{gao2018multi}, 
battery-charge estimation \cite{he2016real}, and
state and parameter estimation in atmospheric models \cite{gove2006application}.

Like the Ensemble Kalman filter (EnKF) \cite{anderson2001ensemble}, UKF propagates an ensemble in order to compute the mean and covariance of the state estimate. 
However, unlike EnKF, which approximates the covariance using statistics of the propagated ensembles, UKF uses unscented transformations to approximate the covariances, which allows UKF to reduce the size of the ensemble to $2n+1,$ where $n$ is the dimension of the state of the system.   
Since UKF propagates the ensemble using the nonlinear dynamics map, the accuracy of UKF is expected and is also reported to be better than that of the Extended Kalman filter, which is based on the linearized dynamics \cite{st2004comparison}.


The UKF gain and covariance update are motivated by the corresponding expressions used in the Kalman filter. 
Hence, it is reasonable to expect that, in the case of a linear system, the UKF gain and the covariance update will coincide with Kalman filter. 
However, it turns out that UKF does not specialize to the classical Kalman filter when applied to a linear system. 
This is due to the fact that effect of the process noise does not pass through to the output-error covariance.
In fact, UKF output covariances and the propagated state covariance are found to be missing the process noise term when applied to a linear system, as shown in this paper.

This paper presents two extension of UKF that specialize to the Kalman filter for linear systems. 
The first extension, named Extended UKF-A (EUKF-A), uses the gradient of the dynamics map to account for the missing term, 
whereas the second extension, named Extended UKF-C (EUKF-A), uses the gradient of the measurement map to account for the missing term.
In the case of a linear system, both of these modifications are equivalent and exactly recover the Kalman filter. 
Note that the names EUKF-A and EUKF-C are motivated by the fact that these modifications use the gradient of the dynamics and the measurement map, similar to EKF.  
However, unlike EKF, EUKF-A and EUKF-C use a $2n+1$-member ensemble along with the gradient of the dynamics map and the measurement map to propagate uncertainty. 
This additional information allows EUKF-A and EUKF-C to improve the accuracy in comparison to UKF.

Since EUKF-A uses the gradient of the dynamics map and requires the computation of its inverse, the improved accuracy might not justify the additional computational cost. 
In contrast, EUKF-C uses the gradient of the measurement map, which in most applications is linear and constant, or computationally inexpensive to compute since the number of outputs is usually much smaller than the dimension of the state, and thus is potentially a low-cost extension of UKF.
Assuming that EnKF gives the true propagated covariance, the nonlinear examples considered in this paper show that both EUKF-A and EUKF-C improve the accuracy of the propagated covariance compared to classical UKF. 



This paper is organized as follows.
Section \ref{sec:summaryKF} briefly reviews the Kalman filter to introduce the terminology and notation used in this paper. 
Section \ref{sec:summaryUKF} briefly reviews UKF.
Section \ref{sec:SpecializationofUKF} applies UKF to a linear system and shows that UKF is suboptimal. 
Section \ref{sec:EUKF} proposes two extensions to the classical UKF that special to Kalman filter in the case of linear systems. 
Section \ref{sec:exmple} applies the proposed extensions to two nonlinear systems and compares the accuracy of uncertainty propagation with EKF, EnKF, and UKF. 
Finally, the paper concludes with a discussion in Section \ref{sec:conclusion}.

\section{SUMMARY OF THE KALMAN FILTER}
\label{sec:summaryKF}
This section briefly reviews the Kalman filter to introduce terminology and notation for later sections.
Consider a linear system
\begin{align}
    x_{k+1}
        &= 
            A_k x_{k} + B_k u_{k} + w_{k },
    \label{eq:LS_xkp1}
    \\
    y_{k} 
        &=
            C_k x_{k} + v_{k },
    \label{eq:LS_yk}
\end{align}
where, for all $k \ge 0$, 
$A_k, B_k, C_k$ are real matrices,
$w_k \sim \SN(0,Q_k)$ is the disturbance, and 
$v_k \sim \SN(0,R_k)$ is the sensor noise.  

For the system \eqref{eq:LS_xkp1}, \eqref{eq:LS_yk}, consider the filter 
%
%
\begin{align}
    \hat x_{k+1|k}
        &=
            A_k \hat x_{k|k} + B_k u_k,
    \label{eq:KF_x(k+1|k)}
        \\
    \hat x_{k+1|k+1}
        &=
            \hat x_{k+1|k} +
              \hat K
            (y_{k+1} - C_{k+1} \hat x_{k+1|k}),
    \label{eq:KF_x(k+1|k+1)}
\end{align}
where 
$\hat x_{k+1|k}$ is the \textit{prior estimate} at step $k+1,$
$\hat x_{k+1|k+1}$ is the \textit{posterior estimate} at step $k+1,$ and 
the gain ${\hat K} \in \BBR^{l_\eta \times l_y}$ is determined by optimization below.  

For all $k\ge0,$ define the \textit{prior error} $e_{k+1|k}$ and 
the \textit{posterior error} $e_{k|k}$ by
\begin{align}
    e_{k+1|k}
        &\isdef
            x_{k+1} - \hat x_{k+1|k},
    \label{eq:e_prior}
    \\
    e_{k+1|k+1}
        &\isdef
            x_{k+1} - \hat x_{k+1|k+1},  
    \label{eq:e_posterior}
\end{align}
and the covariances of $e_{k+1|k}$ and $e_{k+1|k+1}$ by
\begin{align}
    P_{k+1|k}
        &\isdef
            \BBE[e_{k+1|k}   e_{k+1|k}^\rmT],
    \label{eq:P_k_km1_def_2SOF}
    \\
    P_{k+1|k+1}
        &\isdef 
            \BBE[e_{k+1|k+1}   e_{k+1|k+1}^\rmT].
    \label{eq:P_k_k_def}
\end{align}
Note that, for all $k\ge0,$
    \begin{align}
    P_{k+1|k}
        &=
            A_k P_{k|k} A_k^\rmT + Q_k,
    \label{eq:2SOF_prior_cov}
        \\
    P_{k+1|k+1}
        &=
            P_{k+1|k} +
              \hat K
            \overline R_{k+1}
            \hat K^\rmT  
            \neweqline
            -   \hat K C_{k+1} P_{k+1|k}
            - P_{k+1|k} C_{k+1}^\rmT \hat K^\rmT  ,
    \label{eq:2SOF_post_cov}
\end{align}
where 
\begin{align}
    \overline R_{k+1}
        \isdef 
            C_{k+1} P_{k+1|k} C_{k+1}^\rmT + R_{k+1}.
\end{align}

The Kalman gain $K_{k+1 }^{\rm KF}$, defined by
\begin{align}
    K_{k+1 }^{\rm KF}
        \isdef
            \underset{ \hat K \in \BBR^{l_\eta \times l_y}  }{\operatorname{argmin}} \
            {\rm tr} \ P_{k+1|k+1},
\end{align}
is given by 
\begin{align}
        K_{k+1 } ^{\rm KF}
            &=
                P_{k+1|k} C_{k+1}^\rmT
                \overline R_{k+1}^{-1} ,
        \label{eq:2SOF_minimizer}
    \end{align}
and the corresponding optimized posterior covariance at step $k+1$ is given by
    \begin{align}
        &P_{k+1|k+1}
            =
                P_{k+1|k} 
                -
                P_{k+1|k} C_{k+1}^\rmT
                \overline R_{k+1}^{-1}
                C_{k+1} P_{k+1|k} .
           \label{eq:2SOF_optimized_post_cov}
    \end{align}
The \textit{Kalman filter} is \eqref{eq:KF_x(k+1|k)}, \eqref{eq:KF_x(k+1|k+1)} with $\hat K = K_{k+1}^{\rm KF}$, where $ K_{k+1}^{\rm KF}$ is given by 
\eqref{eq:2SOF_prior_cov},
\eqref{eq:2SOF_minimizer}, and \eqref{eq:2SOF_optimized_post_cov}.
    



%
%

Next, 
in order to motivate UKF, 
\eqref{eq:2SOF_post_cov}, \eqref{eq:2SOF_minimizer}, and \eqref{eq:2SOF_optimized_post_cov} 
are reformulated in terms of covariance matrices. 
For all $k\ge0$, define {\it prior output error} $z_{k+1|k}$ and the {\it posterior output error} $z_{k+1|k+1}$  by
\begin{align}
    z_{k+1|k}
        &\isdef
            C_{k+1} e_{k+1|k},
        \\
    z_{k+1|k+1} 
        &\isdef 
            C_{k+1} e_{k+1|k} + v_{k+1}.
\end{align}
%
%
Next, define the covariance of $z_{k+1|k+1}$ and
the cross-covariance of $e_{k+1|k}$ and $z_{k+1|k}$ by
\begin{align}
    P_{z_{k+1|k+1}}
        &\isdef
            \BBE[z_{k+1|k+1} z_{k+1|k+1}^\rmT],
    \label{eq:Pz_k_k_def_2SOF}
    \\
    P_{e,z_{k+1|k}}   
        &\isdef
            \BBE[e_{k+1|k}  z_{k+1|k}^\rmT],
    \label{eq:Pez_k_km1_def_2SOF}
\end{align}
which, for all $k\ge0$, satisfy
\begin{align}
    P_{z_{k+1|k+1}}
        &=
            C_{k+1} P_{k+1|k} C_{k+1}^\rmT + R_{k+1}, 
    \label{eq:2SOF_Pz}
        \\
    P_{e,z_{k+1|k}}
        &=
            P_{k+1|k} C_{k+1}^\rmT.
    \label{eq:2SOF_Pez}
\end{align}
Substituting \eqref{eq:2SOF_Pz} and \eqref{eq:2SOF_Pez} into \eqref{eq:2SOF_post_cov}, the posterior covariance at step $k+1$ can be written as
\begin{align}
    P_{k+1|k+1}
        &=
            P_{k+1|k} +
              \hat K
            P_{z_{k+1|k+1}} 
            \hat K^\rmT  
            \neweqline
            -   \hat K P_{e,z_{k+1|k}}^\rmT
            - P_{e,z_{k+1|k}} \hat K^\rmT  ,
    \label{eq:P_kp1_kp1_IC-UKF}
\end{align}
and
substituting 
\eqref{eq:2SOF_Pz} and
\eqref{eq:2SOF_Pez} in \eqref{eq:2SOF_minimizer} and
\eqref{eq:2SOF_optimized_post_cov},
the Kalman gain can be written as
\begin{align}
    K_{k+1 } ^{\rm KF}
            &=
                P_{e,z_{k+1|k}} 
                P_{z_{k+1|k+1}} ^{-1} ,
    \label{eq:KF_opt_gain_cov}
\end{align}
and the corresponding optimized posterior covariance at step $k+1$ can be written as
\begin{align}
    P_{k+1|k+1}^{\rm KF}
        &=
            P_{k+1|k} -
            K_{k+1 } ^{\rm KF}
            P_{e,z_{k+1|k}} ^\rmT
            .
    \label{eq:KF_opt_cov_cov}
\end{align}

\section{SUMMARY OF UKF}
\label{sec:summaryUKF}
Consider a system
\begin{align}
    x_{k+1}
        &= 
            f_k (x_{k}, u_{k}) + w_{k },
    \label{eq:NLS_xkp1}
    \\
    y_{k} 
        &=
            g_k( x_{k}) + v_{k },
    \label{eq:NLS_yk}
\end{align}
where, for all $k \ge 0$, 
$f_k, g_k, C_k$ are real-valued vector functions,
$w_k \sim \SN(0,Q_k)$ is the disturbance, and 
$v_k \sim \SN(0,R_k)$ is the sensor noise.  


Let $K_{k+1}^{\rm UKF}$ and $P_{k+1|k+1}^{\rm UKF}$ denote the filter gain and the posterior covariance computed by UKF.
In order to compute $K_{k+1}^{\rm UKF}$
and $P_{k+1|k+1}^{\rm UKF},$
UKF approximates the covariance matrices  
$P_{k+1|k},$
$P_{z_{k+1|k+1}},$ and
$P_{e,z_{k+1|k}}$ in \eqref{eq:KF_opt_gain_cov} and \eqref{eq:KF_opt_cov_cov} by propagating an ensemble of $2l_x+1$ sigma points. 
For all $k \ge 0$, the $i$th sigma point $\hat x_{\sigma_i,k}$ is defined as the $i$th column of the ${l_x} \times (2 {l_x}+1)$ matrix
\begin{align}
    X_{k}
        \isdef
            [
                \hat x_{k|k}  \
                &\hat x_{k|k} + p_1 \
                \cdots \
                \hat x_{k|k} + p_{l_\eta} 
                \nn \\ &\quad
                \hat x_{k|k} - p_1 \
                \cdots \
                \hat x_{k|k} - p_{l_\eta} 
            ],
    \label{eq:Xk_sigma_matrix}
\end{align}
where $p_i$ is the $i$th column of 
\begin{align}
    P_{\sigma ,k}
        \isdef
            \alpha  \sqrt{ l_x P_{k|k}^{\rm UKF}},
    \label{eq:Psigma}
\end{align}
where $\alpha\in(0,\infty)$, and $P_{k|k}^{\rm UKF}$ is the approximation of the posterior covariance given by UKF at step $k$.
Then, for all $i = 1, \ldots, 2 l_x+1,$ the sigma points are propagated as 
\begin{align}
    \hat x_{\sigma_i,k+1} 
        &=
            f_k(\hat x_{\sigma_i,k}, u_k), 
    \label{eq:x_sigma_propagated_FS}        
\end{align}
and the corresponding outputs are given by
\begin{align}
    \hat y_{\sigma_i,k+1} 
        &=
            g_{k+1}(\hat x_{\sigma_i}   ).
\end{align}

Defining 
\begin{align}
    X_{k+1|k}
        \isdef
            \matl{ccc}
                \hat x_{\sigma_1,k+1}   & \cdots    & \hat x_{\sigma_{2l_x+1},k+1} 
            \matr
            \in \BBR^{l_x \times 2 l_x+1},
    \\
    Y_{k+1}
        \isdef
            \matl{ccc}
                \hat y_{\sigma_1,k+1}   & \cdots    & \hat y_{\sigma_{2l_x+1},k+1} 
            \matr
            \in \BBR^{l_y \times 2 l_x+1},
\end{align}
the covariance matrices in \eqref{eq:KF_opt_gain_cov} and \eqref{eq:KF_opt_cov_cov} are then approximated by
\begin{align}
    P_{k+1|k}^{\rm UKF}
        &=
            \tilde X_{k+1} W_\rmd  \tilde X_{k+1}^\rmT  + Q_k,
    \label{eq:Px_UKF}
        \\
    P_{z_{k+1|k+1}}^{\rm UKF}
        &=
            \tilde Y_{k+1} W_\rmd \tilde Y_{k+1}^\rmT  + R_{k+1},
    \label{eq:Pz_UKF}
        \\
    P_{e,z_{k+1|k}}^{\rm UKF}
        &=
            \tilde X_{k+1} W_\rmd \tilde Y_{k+1}^\rmT,
    \label{eq:Pez_UKF}
\end{align}
where
\begin{align}
    \tilde X_{k+1}
        &\isdef
            X_{k+1|k} - H(X_{k+1|k} W)
    \label{eq:tilde_X_def}
    , \\
    \tilde Y_{k+1}
        &\isdef
            Y_{k+1} - 
            H(Y_{k+1} W),
    \label{eq:tilde_Y_def}
\end{align}
where, for $ x \in \BBR^n,$
\begin{align}
    H(x) 
        \isdef
            1_{1 \times 2 l_x+1 } \otimes x
            \in \BBR^{2 l_x+1 \times n}
            ,
\end{align}
and
\begin{align}
       W
        \isdef
            \dfrac{1}{2\alpha^2 l_x}
            \matl{c}
                2 (\alpha^2 -1 ) l_x \\ 
                 1_{2 l_x \times 1} \\
            \matr\in\BBR^{2l_x+1},
    \quad
    W_\rmd 
        \isdef 
            {\rm diag \ } W .
    \label{eq:UKF_W_def}
\end{align}
Note that  $X_{k+1} W$ is a weighted average of the propagated sigma points.
Therefore, the entries of $\tilde X_{k+1}$ defined by \eqref{eq:tilde_X_def} are perturbations of the weighted average determined by the propagated sigma points, and
the entries of $\tilde Y_{k+1}$ are the corresponding output perturbations.
%
%
Finally,  UKF filter gain 
and the corresponding posterior covariance are given by
\begin{align}
    K_{k+1 } ^{\rm UKF}
        &=
            P_{e,z_{k+1|k}} ^{\rm UKF}
            P_{z_{k+1|k+1}}^{{\rm UKF} ^{-1}} ,
    \label{eq:UKF_opt_PzPez}
    \\
    P_{k+1|k+1}^{\rm UKF}
        &=
        P_{k+1|k} ^{\rm UKF}
        -
        K_{k+1 } ^{\rm UKF}
        P_{e,z_{k+1|k}}^{{\rm UKF}^\rmT}
        ,
   \label{eq:UKF_P_post_PzPez}
\end{align}
and the prior estimate $\hat x_{k+1|k}$ and posterior estimate $\hat x_{k+1|k+1}$ are given by
\begin{align}
    \hat x_{k+1|k}
        &=
            X_{k+1|k} W,
        \label{eq:UKF_prior_update}
        \\
    \hat x_{k+1|k+1}
        &=
            \hat x_{k+1|k} +
            K_{k+1 } ^{\rm UKF}
            (y_{k+1} - Y_{k+1} W).
        \label{eq:UKF_post_update}
\end{align}
Note that \eqref{eq:UKF_opt_PzPez} and \eqref{eq:UKF_P_post_PzPez} are similar in form to \eqref{eq:KF_opt_gain_cov} and \eqref{eq:KF_opt_cov_cov}.

\section{SPECIALIZATION OF UKF TO LINEAR SYSTEMS}
\label{sec:SpecializationofUKF}
The following result shows that UKF does not specialize to the Kalman filter when applied to a linear system.

\begin{proposition}
    \label{prop:apply_UKF_to_LS}
    Consider a linear system \eqref{eq:LS_xkp1}, \eqref{eq:LS_yk}.
    For all $k\ge0,$ let $P_{k+1|k+1}$ be the posterior covariance given by Kalman filter
    and 
    let $P_{k+1|k+1}^{\rm UKF}$ be the posterior covariance given by  UKF.
    Let $k\ge 0,$ and assume that 
    \begin{align}
        P_{k|k}=P_{k|k}^{\rm UKF} 
        ,
        \label{eq:Pbar_assumption}
    \end{align}
    $Q_k \neq 0,$ and $C_k \notin \SN(Q_k).$
    Then, 
    \begin{align}
        P_{k+1|k+1}^{\rm UKF}
            \neq 
                P_{k+1|k+1}.
        \label{eq:UKFP_neq_KFP}
    \end{align}

    Furthermore, denote the posterior covariance at step $k+1$ obtained with gain ${\hat K} \in \BBR^{l_\eta \times l_y}$ by
    \begin{align}
        P(\hat K) 
            &=
                P_{k+1|k} +
              \hat K
            P_{z_{k+1|k+1}}^\rmT
            \hat K^\rmT  
            \neweqline
            -   \hat K P_{e,z_{k+1|k}}^\rmT
            - P_{e,z_{k+1|k}} \hat K^\rmT  .
        \label{eq:PatK}
    \end{align}
    Then, 
    \begin{align}
        P_{k+1|k+1}^{\rm UKF}
            \neq 
                P(K_{k+1}^{\rm UKF}), 
        \label{eq:UKFP_neq_PatUKF}
    \end{align}
    and
    \begin{align}
        \tr P(K_{k+1}^{\rm KF})
            \leq 
                \tr P(K_{k+1}^{\rm UKF}).
        \label{eq:traceKF_leq_traceUKF}
    \end{align}

\end{proposition}

\begin{proof}
Note that, for $i = 1, \ldots, 2 l_x+1,$ 
\begin{align}
    \hat x_{\sigma_i,k+1} 
        &=
            A_k \hat x_{\sigma_i,k} + B_k u_{k}, 
        \nn 
    \\
    \hat y_{\sigma_i,k+1} 
        &=
            C_{k+1} \hat x_{\sigma_i,k+1},
        \nn
\end{align}
and thus
\begin{align}
    X_{k+1|k}
        &=
            A_k X_k + H(  B_k u_k ),
        \nn \\
    Y_{k+1} 
        &= C_{k+1} X_{k+1|k}. \nn
\end{align}
Next, noting that $X_k W = \hat x_{k|k}$ and sum of entries of $W$ is one, it follows that
\begin{align}
    X_{k+1|k} W
        &=
            A_k \hat x_{k|k} + B_k u_k ,
        \nn \\
    Y_{k+1} W 
        &=
            C_{k+1} A_k \hat x_{k|k} + C_{k+1} B_k u_k   ,
        \nn 
\end{align}
and thus
\begin{align}
    \tilde X_{k+1} 
        &=
            A_k X_k 
            - H( A_k \hat x_{k|k} )
        \nn 
        \\
        &=
            A_k 
            \matl{ccc}
                0 &
                \alpha\sqrt{l_x P_{k|k}^{\rm UKF}  } &
                -\alpha\sqrt{l_x P_{k|k}^{\rm UKF}}
            \matr,
    \label{eq:tildeX_lin_sys}
        \\
    \tilde Y_{k+1} 
        &=
            Y_{k+1} - H(Y_{k+1}W)
        \nn
        \\
        &=
            C_{k+1} \tilde X_{k+1}.
        \label{eq:tildeY_lin_sys}
\end{align}

Using \eqref{eq:2SOF_prior_cov}, \eqref{eq:Pbar_assumption} and \eqref{eq:tildeX_lin_sys}, it follows from \eqref{eq:Px_UKF} that 
%
\begin{align}
    P_{k+1|k}^{\rm UKF}
        &=
            A_k 
            \matl{ccc}
                0 &
                \alpha\sqrt{l_x P_{k|k}} &
                -\alpha\sqrt{l_x P_{k|k}}
            \matr
            W_\rmd
            \nn \\ &\quad 
            \cdot \matl{ccc}
                0 &
                \alpha\sqrt{l_x P_{k|k}} &
                -\alpha\sqrt{l_x P_{k|k}}
            \matr^\rmT 
            A_k^\rmT 
            + Q_k
        \nn \\
        &=
            A_k 
            P_{k|k}
            A_k^\rmT 
            + Q_k
            \nn
        \\
        &=
            P_{k+1|k} \nn.
\end{align}
Using \eqref{eq:2SOF_prior_cov}, \eqref{eq:2SOF_Pz} and \eqref{eq:tildeY_lin_sys}, it follows from \eqref{eq:Pz_UKF} that 
\begin{align}
    P_{z_{k+1|k+1}}^{\rm UKF}
        &=
            C_{k+1} \tilde X_{k+1} W_\rmd \tilde X_{k+1}^\rmT C_{k+1}^\rmT  + R_{k+1}
            \nn
        \\
        &=
            C_{k+1} A_k P_{k|k}  A_k^\rmT C_{k+1}^\rmT + R_{k+1}
            \nn
        \\
        &=
            C_{k+1} (P_{k+1|k} - Q_k ) C_{k+1}^\rmT + R_{k+1}
            \nn
        \\
        &=
            C_{k+1} P_{k+1|k} C_{k+1}^\rmT 
            - C_{k+1} Q_k C_{k+1}^\rmT 
            + R_{k+1}
            \nn
        \\
        &=
            P_{z_{k+1|k+1}} - C_{k+1} Q_k C_{k+1}^\rmT .
        \label{eq:UKF_Pz}
\end{align}
Using \eqref{eq:2SOF_prior_cov} and \eqref{eq:2SOF_Pez}, it follows from \eqref{eq:Pez_UKF} that 
\begin{align}
    P_{e,z_{k+1|k}}^{\rm UKF}
        &=
            \tilde X_{k+1} W_\rmd \tilde Y_{k+1}^\rmT
        \nn 
        \\
        &=
            \tilde X_{k+1} W_\rmd \tilde X_{k+1}^\rmT C_{k+1}^\rmT 
        \nn 
        \\
        &=
            A_k P_{k|k} A_k^\rmT C_{k+1}^\rmT 
        \nn 
        \\
        &=
            P_{k+1|k} C_{k+1}^\rmT - Q_k C_{k+1}^\rmT 
        \nn
        \\
        &=
            P_{e,z_{k+1|k}} - Q_k C_{k+1}^\rmT .
    \label{eq:UKF_Pez}
\end{align}
Since $Q_k \neq 0$ and $C_k \notin \SN(Q_k),$ it follows that 
$P_{z_{k+1|k+1}}^{\rm UKF} \neq P_{z_{k+1|k+1}}$ and 
$P_{e,z_{k+1|k}}^{\rm UKF} \neq P_{e,z_{k+1|k}}$
are missing $C_{k+1} Q_k C_{k+1}^\rmT$ and $Q_k C_{k+1}^\rmT$, respectively, 
, thus implying \eqref{eq:UKFP_neq_KFP}.
Next, substituting \eqref{eq:UKF_Pz} and \eqref{eq:UKF_Pez} in \eqref{eq:UKF_P_post_PzPez} proves \eqref{eq:UKFP_neq_PatUKF}.

To prove \eqref{eq:traceKF_leq_traceUKF}, note that 
\begin{align}
    K_{k+1 } ^{\rm UKF}
        &=
            (P_{e,z_{k+1|k}} - Q_k C_{k+1}^\rmT)
            \nn \\ &\quad 
            \cdot (P_{z_{k+1|k+1}} - C_{k+1} Q_k C_{k+1}^\rmT) ^{-1} ,
        \nn \\
        &\neq 
            K_{k+1}^{\rm KF}.
        \label{eq:KKFneqKUKF}
\end{align}
Finally, since $K_{k+1}^{\rm KF}$ minimizes $\tr P_{k+1|k+1}$,
\eqref{eq:KKFneqKUKF} implies \eqref{eq:traceKF_leq_traceUKF}.
%
%
\end{proof}


Note that, in a linear system, the UKF prior and posterior updates given by \eqref{eq:UKF_prior_update} and \eqref{eq:UKF_post_update} reduce to \eqref{eq:KF_x(k+1|k)} and \eqref{eq:KF_x(k+1|k+1)}, where $\hat K = K_{k+1}^{\rm UKF}.$
However, Proposition \ref{prop:apply_UKF_to_LS} implies that, in a linear system  where disturbance is not zero, UKF does not reduce to Kalman filter.
That is, the posterior covariance propagated by UKF is not equal to the covariance defined by \eqref{eq:P_k_k_def}.
%
%
Finally, note that, in linear systems, the choice of $\alpha$ does not affect $K_{k+1}^{\rm UKF}$ and $P_{k+1|k+1}^{\rm UKF}.$

Furthermore, the covariance corresponding to the gain $K_{k+1}^{\rm UKF}$  is, in fact, $P(K_{k+1}^{\rm UKF})$ given by \eqref{eq:PatK}, which is not equal to $P_{k+1|k+1}^{\rm UKF}.$ 
As shown in the next example, $\tr P_{k+1|k+1}^{\rm UKF}$ can be smaller than $\tr P(K_{k+1}^{\rm KF}),$ which is impossible. 
This apparent contradiction is due to the fact that UKF uses incorrect equation to update the posterior covariance. 

\begin{exmp}
    Consider a linear system where, for all $k\ge 0,$ 
    \begin{align}
        A_k
            &=
                \matl{cc}
                    2.4 & 2.1 \\
                    0 & -0.7
                \matr, 
        C_k 
            =
                \matl{cc}
                    -0.4 \ -0.9
                \matr,
    \end{align}
    $Q_k = 1,$ and  $R_k=1.$
    Let $x_0 = [1 \ 1]^\rmT$ and  $P_{0|0} = I_2$.
    Note that ${\rm mspec} (A) = \{2.4, -0.7 \}$ and $(A,C)$ is detectable. 
    In this case, 
    \begin{align}
        &\tr P_{1|1} = \tr P(K_{1}^{\rm KF})     = 9.079, \\
        &\tr P_{1|1}^{\rm UKF}     = 8.816, \\
        &\tr P(K_{1}^{\rm UKF})     = 9.730.
    \end{align}
    Note that the trace of UKF posterior covariance is smaller than the trace of KF posterior covariance, which is clearly a contradiction, since posterior covariance given by Kalman filter is optimal. 
    The true covariance corresponding to the UKF gain is in fact larger than the KF posterior covariance. 
    $\hfill\mbox{\Large$\diamond$}$
\end{exmp}

\begin{exmp}
    \label{exmp:lin_exmp}
    Consider a linear system where, for all $k\ge 0$ 
    \begin{align}
        A_k
            &=
                \matl{cc}
                    1.6 & -1 \\
                    1 & 0
                \matr, 
        C_k 
            =
                \matl{cc}
                    1 \ -0.3
                \matr,
    \end{align}
    $Q_k = 0.1,$ and  $R_k=0.1.$
    Let $x_0 = [1 \ 1]^\rmT$ and  $P_{0|0} = I_2$.
    Note that ${\rm mspec} (A) = \{0.8+0.6\jmath, 0.8-0.6\jmath \}$ and $(A,C)$ is detectable. 
    Figure shows the trace of 
    $P_{k|k} = \tr P(K_{k}^{\rm KF}) $ and 
    $ P_{k|k}^{\rm UKF}.$
    Clearly, for $k>1,$ $P_{k|k} \neq P_{k|k}^{\rm UKF}.$
    $\hfill\mbox{\Large$\diamond$}$
\end{exmp}

\begin{figure}[h]
    \centering
    \includegraphics[width = 1\columnwidth]{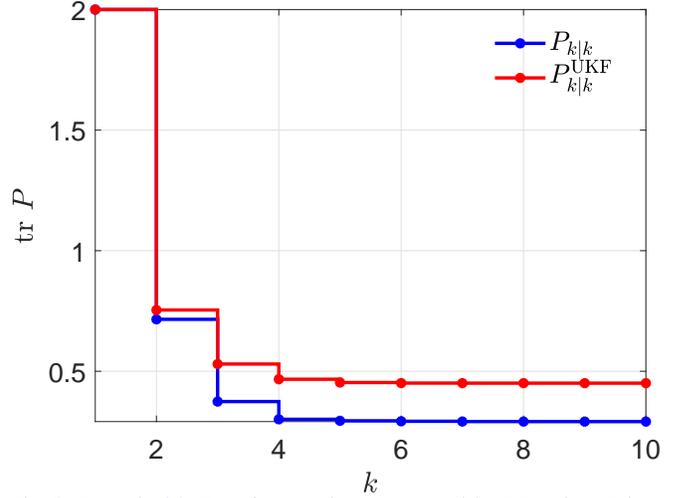}
    \caption{
    Example \ref{exmp:lin_exmp}.
    Posterior covariance computed by KF and UKF in a linear system. 
    Note that, for $k>1,$ $P_{k|k} \neq P_{k|k}^{\rm UKF}.$
    }
    \label{fig:MUKF_Exmp2}
\end{figure}


\section{TWO MODIFICATIONS TO UKF}
\label{sec:EUKF}

As shown in the previous section, the covariances $P_{z_{k+1|k+1}}$ and $P_{e,z_{k+1|k}}$ in \eqref{eq:UKF_Pz} and \eqref{eq:UKF_Pez}  are missing terms that depend on the disturbance statistics $Q_k$, thus preventing UKF from specializing to the Kalman filter for linear systems.
To remedy this omission, this section presents two modifications of the UKF algorithm,
namely
Extended UKF-A (EUKF-A) and 
Extended UKF-C (EUKF-C),
both of which specialize to the Kalman filter for linear systems.
In both of these modification, the UKF covariance matrices \eqref{eq:Px_UKF}-\eqref{eq:Pez_UKF} are modified such that they specialize to \eqref{eq:2SOF_prior_cov}, \eqref{eq:2SOF_Pz}, and \eqref{eq:2SOF_Pez} in the case of linear systems.

\begin{table*}[t]
    \centering
    \begin{tabular}{|c|c|c|c|}
        \hline
        Variable & UKF & EUKF-A & EUKF-C \\
        \hline
            $P_{\sigma,k}$ &
            $\alpha  \sqrt{ l_x (P_{k|k}^{\rm UKF})}$ &
            $\alpha  \sqrt{ l_x (P_{k|k}^{\rm EUKFA}+A_k^{-1} Q_k A_k^{\rmT^{-1}})}$ & 
            $\alpha  \sqrt{ l_x (P_{k|k}^{\rm EUKFC})}$  
            \\
        & \eqref{eq:Psigma} & \eqref{eq:Psigma_MUKF1} & \eqref{eq:Psigma_MUKF2} 
            \\
        \hline
            $P_{k+1|k}$ &
            $\tilde X_{k+1} W_\rmd  \tilde X_{k+1}^\rmT  + Q_k$ &
            $\tilde X_{k+1} W_\rmd  \tilde X_{k+1}^\rmT$ &
            $\tilde X_{k+1} W_\rmd  \tilde X_{k+1}^\rmT  + Q_k$ 
            \\
        & \eqref{eq:Px_UKF} & \eqref{eq:Px_MUKF1} & \eqref{eq:Px_MUKF2}
            \\
        \hline
            $P_{ez,_{k+1|k}}$ &
            $\tilde X_{k+1} W_\rmd  \tilde Y_{k+1}^\rmT  + R_{k+1}$ &
            $\tilde X_{k+1} W_\rmd  \tilde Y_{k+1}^\rmT  + R_{k+1}$ &       
            $\tilde X_{k+1} W_\rmd  \tilde Y_{k+1}^\rmT  + C_{k+1} Q_k C_{k+1}^\rmT + R_{k+1}$ 
            \\
        & \eqref{eq:Pz_UKF} & \eqref{eq:Pz_MUKF1} & \eqref{eq:Pz_MUKF2}
            \\
        \hline
            $P_{z_{k+1|k+1}}$ &  
            $\tilde Y_{k+1} W_\rmd  \tilde Y_{k+1}^\rmT$ &
            $\tilde Y_{k+1} W_\rmd  \tilde Y_{k+1}^\rmT$ &
            $\tilde Y_{k+1} W_\rmd  \tilde Y_{k+1}^\rmT + Q_k C_{k+1}^\rmT$ 
            \\
        & \eqref{eq:Pez_UKF} & \eqref{eq:Pez_MUKF1} & \eqref{eq:Pez_MUKF2}
            \\
        \hline
    \end{tabular}
    \caption{Covariance matrices used in UKF, EUKF-A, and EUKF-C.}
    \label{tab:UKF_mods}
\end{table*}

\subsection{Extended UKF-A}

Assuming, for all $k\ge 0,$ $A_k$ is nonsingular, EUKF-A modifies the sigma points to account for the missing terms in \eqref{eq:Px_UKF}-\eqref{eq:Pez_UKF} as shown below.  
Letting $P_{k|k}^{\rm EUKFA}$ denote the posterior covariance at step $k$, 
define 
\begin{align}
    P_{\sigma,k} ^{\rm EUKFA}
        \isdef 
            \alpha  \sqrt{ l_x (P_{k|k}^{\rm EUKFA}+A_k^{-1} Q_k A_k^{\rmT^{-1}})}
            .
    \label{eq:Psigma_MUKF1}
\end{align}
The sigma points in EUKF-A are then given by \eqref{eq:Xk_sigma_matrix},  
where $p_i$ is the $i$th column of $P_{\sigma,k} ^{\rm EUKFA}$.
With the modified sigma points, define 
\begin{align}
    P_{k+1|k}^{\rm EUKFA}
        &\isdef
            \tilde X_{k+1} W_\rmd  \tilde X_{k+1}^\rmT , 
    \label{eq:Px_MUKF1}
        \\
    P_{z_{k+1|k+1}}^{\rm EUKFA}
        &\isdef
            \tilde Y_{k+1} W_\rmd \tilde Y_{k+1}^\rmT  + R_{k+1},
    \label{eq:Pz_MUKF1}
        \\
    P_{e,z_{k+1|k}}^{\rm EUKFA}
        &\isdef
            \tilde X_{k+1} W_\rmd \tilde Y_{k+1}^\rmT,
    \label{eq:Pez_MUKF1}
\end{align}
where 
$\tilde X_{k+1}$ and
$\tilde Y_{k+1}$ are given by \eqref{eq:tilde_X_def} and \eqref{eq:tilde_Y_def}. 
The filter gain and the posterior covariance are given by
\begin{align}
    K_{k+1 } ^{\rm EUKFA}
        &=
            P_{e,z_{k+1|k}} ^{\rm EUKFA}
            P_{z_{k+1|k+1}}^{{\rm EUKFA} ^{-1}} ,
    \label{eq:MUKF1_opt_PzPez}
    \\
    P_{k+1|k+1}^{\rm EUKFA}
        &=
        P_{k+1|k} ^{\rm EUKFA}
        -
        P_{e,z_{k+1|k}}^{\rm EUKFA}
        P_{z_{k+1|k+1}}^{{\rm EUKFA}^{-1}}
        P_{e,z_{k+1|k}}^{{\rm EUKFA}^\rmT}
        .
   \label{eq:MUKF1_P_post_PzPez}
\end{align}
Finally, 
the prior estimate $\hat x_{k+1|k}$ is given by \eqref{eq:UKF_prior_update} and 
the posterior estimate $\hat x_{k+1|k+1}$ is given by
\begin{align}
    \hat x_{k+1|k+1}
        &=
            \hat x_{k+1|k} +
            K_{k+1 } ^{\rm EUKFA}
            (y_{k+1} - Y_{k+1} W).
        \label{eq:MUKF1_post_update}
\end{align}

The next result shows that EUKF-A reduces to KF in the case of a linear system.  
\begin{proposition}
    \label{prop:MUKF1_is_KF}
    Consider a linear system \eqref{eq:LS_xkp1}, \eqref{eq:LS_yk}.
    For all $k\ge0,$ let $P_{k+1|k+1}$ be the posterior covariance given by Kalman filter
    and 
    let $P_{k+1|k+1}^{\rm EUKFA}$ be the posterior covariance given by EUKF-A.
    Let $k\ge 0,$ and assume that 
    \begin{align}
        P_{k|k}=P_{k|k}^{\rm EUKFA} 
        .
        \label{eq:PMUKF1_assumption}
    \end{align}
    Then, 
    \begin{align}
        K_{k+1}^{\rm EUKFA}
            &= 
                K_{k+1},
        \label{eq:MUKF1K_eq_KFP}
            \\
        P_{k+1|k+1}^{\rm EUKFA}
            &= 
                P_{k+1|k+1}.
        \label{eq:MUKF1P_eq_KFP}
    \end{align}
\end{proposition}

\begin{proof}
Note that 
\begin{align}
    \tilde X_{k+1} 
        &=
            A_k 
            \Bigg[
                0 \
                \alpha\sqrt{l_x (P_{k|k}^{\rm EUKFA}+A_k^{-1} Q_k A_k^{\rmT^{-1}})} 
            \nn \\ &\quad \quad \quad \quad 
                -\alpha\sqrt{l_x (P_{k|k}^{\rm EUKFA}+A_k^{-1} Q_k A_k^{\rmT^{-1}})}
            \Bigg]
    \nn
\end{align}
and thus
\begin{align}
    P_{k+1|k}^{\rm EUKFA}
        &=
            A_k
            (P_{k|k}^{\rm EUKFA}+A_k^{-1} Q_k A_k^{\rmT^{-1}})
            A_k^\rmT
        \nn 
        \\
        &=
            A_k
            P_{k|k}^{\rm EUKFA}
            A_k^\rmT
            +
            Q_k
        \nn 
        \\
        &=
            P_{k+1|k},
        \label{eq:MUKF1_P_is_KF}
        \\
    P_{z_{k+1|k+1}}^{\rm EUKFA}
        &=
            C_{k+1} \tilde X_{k+1} W_\rmd \tilde X_{k+1}^\rmT C_{k+1}^\rmT  + R_{k+1}
            \nn
        \\
        &=
            C_{k+1} (A_k
            P_{k|k}^{\rm EUKFA}
            A_k^\rmT
            +
            Q_k) C_{k+1}^\rmT + R_{k+1}
            \nn
        \\
        &=
            C_{k+1} P_{k+1|k}  C_{k+1}^\rmT + R_{k+1}
            \nn
        \\
        &=
            P_{z_{k+1|k+1}} ,
        \label{eq:MUKF1_Pz_is_KF}
    \\
    P_{e,z_{k+1|k}}^{\rm UKF}
        &=
            \tilde X_{k+1} W_\rmd \tilde Y_{k+1}^\rmT
        \nn 
        \\
        &=
            \tilde X_{k+1} W_\rmd \tilde X_{k+1}^\rmT C_{k+1}^\rmT 
        \nn 
        \\
        &=
            P_{k+1|k} C_{k+1}^\rmT 
        \nn
        \\
        &=
            P_{e,z_{k+1|k}}.
    \label{eq:MUKF1_Peq_is_KF}
\end{align}
Equations \eqref{eq:MUKF1_P_is_KF}-\eqref{eq:MUKF1_Peq_is_KF} immidiately imply  \eqref{eq:MUKF1K_eq_KFP} and \eqref{eq:MUKF1P_eq_KFP}.
\end{proof}

\subsection{Extended UKF-C}
Using $C_k,$ EUKF-C adds the missing terms in \eqref{eq:Pz_UKF}, \eqref{eq:Pez_UKF} as shown below.
%
Letting $P_{k|k}^{\rm EUKFC}$ denote the posterior covariance at step $k,$
define the sigma points by \eqref{eq:Xk_sigma_matrix},  
where $p_i$ is the $i$th column of 
\begin{align}
    P_{\sigma,k} ^{\rm EUKFC}
        \isdef 
            \alpha  \sqrt{ l_x P_{k|k}^{\rm EUKFC}}
            .
    \label{eq:Psigma_MUKF2}
\end{align}
Next, define
\begin{align}
    P_{k+1|k}^{\rm EUKFC}
        &\isdef
            \tilde X_{k+1} W_\rmd  \tilde X_{k+1}^\rmT  + Q_k,
    \label{eq:Px_MUKF2}
        \\
    P_{z_{k+1|k+1}}^{\rm EUKFC}
        &\isdef
            \tilde Y_{k+1} W_\rmd \tilde Y_{k+1}^\rmT 
            + C_{k+1} Q_k C_{k+1}^\rmT
            + R_{k+1} ,
    \label{eq:Pz_MUKF2}
        \\
    P_{e,z_{k+1|k}}^{\rm EUKFC}
        &\isdef
            \tilde X_{k+1} W_\rmd \tilde Y_{k+1}^\rmT
            + Q_k C_{k+1}^\rmT
            ,
    \label{eq:Pez_MUKF2}
\end{align}
where 
$\tilde X_{k+1}$ and
$\tilde Y_{k+1}$ are given by \eqref{eq:tilde_X_def} and \eqref{eq:tilde_Y_def}. 
The filter gain and the posterior covariance are given by
\begin{align}
    K_{k+1 } ^{\rm EUKFC}
        &=
            P_{e,z_{k+1|k}} ^{\rm EUKFC}
            P_{z_{k+1|k+1}}^{{\rm EUKFC} ^{-1}} ,
    \label{eq:MUKF2_opt_PzPez}
    \\
    P_{k+1|k+1}^{\rm EUKFC}
        &=
        P_{k+1|k} ^{\rm EUKFC}
        -
        P_{e,z_{k+1|k}}^{\rm EUKFC}
        P_{z_{k+1|k+1}}^{{\rm EUKFC}^{-1}}
        P_{e,z_{k+1|k}}^{{\rm EUKFC}^\rmT}
        .
   \label{eq:MUKF2_P_post_PzPez}
\end{align}
Finally, the prior estimate $\hat x_{k+1|k}$ is given by \eqref{eq:UKF_prior_update} and the posterior estimate $\hat x_{k+1|k+1}$ is given by
\begin{align}
    \hat x_{k+1|k+1}
        &=
            \hat x_{k+1|k} +
            K_{k+1 } ^{\rm EUKFC}
            (y_{k+1} - Y_{k+1} W).
        \label{eq:MUKF2_post_update}
\end{align}

The next result shows that EUKF-C reduces to KF in the case of a linear system.
\begin{proposition}
    \label{prop:MUKF2_is_KF}
    Consider a linear system \eqref{eq:LS_xkp1}, \eqref{eq:LS_yk}.
    For all $k\ge0,$ let $P_{k+1|k+1}$ be the posterior covariance given by Kalman filter
    and 
    let $P_{k+1|k+1}^{\rm EUKFC}$ be the posterior covariance given by EUKF-C.
    Let $k\ge 0,$ and assume that 
    \begin{align}
        P_{k|k}=P_{k|k}^{\rm EUKFC} 
        .
        \label{eq:PMUKF2_assumption}
    \end{align}
    Then, 
    \begin{align}
        K_{k+1}^{\rm EUKFC}
            &= 
                K_{k+1},
        \label{eq:MUKF2K_eq_KFP}
            \\
        P_{k+1|k+1}^{\rm EUKFC}
            &= 
                P_{k+1|k+1}.
        \label{eq:MUKF2P_eq_KFP}
    \end{align}
\end{proposition}

\begin{proof}
Note that 
\begin{align}
    \tilde X_{k+1} 
        &=
            A_k 
            \matl{ccc}
                0 &
                \alpha\sqrt{l_x P_{k|k}^{\rm EUKFC}} 
                &
                -\alpha\sqrt{l_x P_{k|k}^{\rm EUKFC}} 
            \matr
    \nn
\end{align}
and thus
\begin{align}
    \tilde X_{k+1} W_\rmd \tilde X_{k+1} ^\rmT 
        =
            A_k 
            P_{k|k}^{\rm EUKFC}
            A_k^\rmT.
    \label{eq:tildeX_mat_MUKF2}
\end{align}
Substituting \eqref{eq:tildeX_mat_MUKF2} in \eqref{eq:Px_MUKF2}-\eqref{eq:Pez_MUKF2} yields
\begin{align}
    P_{k+1|k}^{\rm EUKFC}
        &=
            A_k
            P_{k|k}^{\rm EUKFC}
            A_k^\rmT
            +
            Q_k
        \nn 
        \\
        &=
            P_{k+1|k},
        \label{eq:MUKF2_P_is_KF}
        \\
    P_{z_{k+1|k+1}}^{\rm EUKFC}
        &=
            C_{k+1} \tilde X_{k+1} W_\rmd \tilde X_{k+1}^\rmT C_{k+1}^\rmT 
            \neweqline
            + C_{k+1} Q_k C_{k+1}^\rmT
            + R_{k+1}
            \nn
        \\
        &=
            C_{k+1} 
            (
                A_k
                P_{k|k}^{\rm EUKFC}
                A_k^\rmT
                +
                Q_k
            )
            C_{k+1}^\rmT 
            + R_{k+1}
            \nn
        \\
        &=
            C_{k+1} P_{k+1|k}  C_{k+1}^\rmT + R_{k+1}
            \nn
        \\
        &=
            P_{z_{k+1|k+1}} ,
        \label{eq:MUKF2_Pz_is_KF}
    \\
    P_{e,z_{k+1|k}}^{\rm UKF}
        &=
            \tilde X_{k+1} W_\rmd \tilde Y_{k+1}^\rmT
            + Q_k C_{k+1}^\rmT
        \nn 
        \\
        &=
            \tilde X_{k+1} W_\rmd \tilde X_{k+1}^\rmT C_{k+1}^\rmT 
            + Q_k C_{k+1}^\rmT
        \nn 
        \\
        &=
            (A_k 
            P_{k|k}^{\rm EUKFC}
            A_k^\rmT
            + Q_k )C_{k+1}^\rmT
        \nn 
        \\
        &=
            P_{k+1|k} C_{k+1}^\rmT 
        \nn
        \\
        &=
            P_{e,z_{k+1|k}}.
    \label{eq:MUKF2_Peq_is_KF}
\end{align}
Equations \eqref{eq:MUKF2_P_is_KF}-\eqref{eq:MUKF2_Peq_is_KF} immidiately imply  \eqref{eq:MUKF2K_eq_KFP} and \eqref{eq:MUKF2P_eq_KFP}.
\end{proof}

\section{NUMERICAL EXAMPLES}
\label{sec:exmple}
In this section,  EUKF-A and EUKF-C  are applied to nonlinear systems. 
In order to compare the performance of EUKF-A and EUKF-C with UKF, EnKF is used to propagate the true posterior covariance.
EKF is also used to estimate the posterior covariance since EUKF-A and EUKF-C are expected to recover the performance of EKF. 



Note that, to apply EKF, EUKF-A, and EUKF-C to nonlinear systems, the dynamics matrix $A_k$ and the output matrix $C_k$ are approximated by
\begin{align}
    A_k 
        &=
            \left. \frac{\partial f}{\partial x} \right|_{\hat x_{k|k}},
        \quad 
    C_k 
        =
            \left. \frac{\partial g}{\partial x} \right|_{\hat x_{k|k}}.
\end{align}

\begin{exmp}
    \label{exmp:VDP}
    \textit{Van der Pol Oscillator}.
Consider the discretized Van der Pol Oscillator. 
\begin{align}
    x_{k+1} 
        &=
            f(x_k) + w_k,
\end{align}
where 
\begin{align}
    f(x)
        =
            \matl{c}
                x_1 + T_\rms x_2 \\
                x_2 + T_\rms (\mu (1-x_1^2) x_2-x_1)
            \matr.
\end{align}
Let the measurement be given by
\begin{align}
    y_k
        =
            C x_k + v_k,
\end{align}
where $C \isdef [1 \ 0].$
For all $k\ge0,$ let $Q_k = 0.01 I_2$  and $R_k = 10^{-4}.$
Furthermore, let $x(0) = [1 \  1 ]^\rmT$ and $P_{0|0} = I_3$.

Letting $\alpha = 1.5$ in UKF, EUKF-A, and EUKF-C, 
Figure \ref{fig:MUKF_VDP_Pk}a) shows the trace of the posterior covariance computed by EnKF, EKF, UKF, EUKF-A, and EUKF-C. 
The true covariance is assumed to be given by EnKF with 100,000 ensemble members. 
Note that UKF overestimates the EnKF posterior covariance, whereas EUKF-A and EUKF-C closely track the EnKF posterior covariance and recover the EKF posterior covariance.  
Figure \ref{fig:MUKF_VDP_Pk}b) shows the error of UKF, EKF, EUKF-A, and EUKF-C posterior covariance relative to EnKF. 
At the end of the simulation, UKF relative error is approximately $15\%$, whereas EUKF-A and EUKF-C relative error is less that $2\%.$
%

Figure \ref{fig:MUKF_VDP_zk_ek} shows the output error $z_{k|k}$ and the norm of the posterior error $e_{k|k}$ computed with all algorithms. 
Note that the output error and the posterior error are very close to each other. 
This example shows that the EUKF-A and EUKF-C posterior covariance estimate is more accurate than the UKF posterior covariance and is approximately equal to the EKF posterior covariance, however, the state estimates computed using all algorithms are almost equal.  
$\hfill\mbox{\Large$\diamond$}$
\end{exmp}

\begin{figure}[h]
    \centering
    \includegraphics[width = 1\columnwidth]{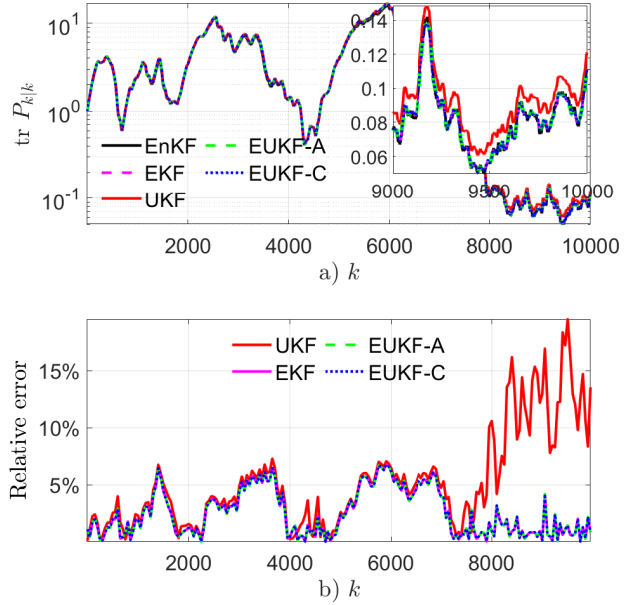}
    \caption{
    Example \ref{exmp:VDP}. Posterior covariance computed using EnKF, EKF, UKF,  EUKF-A, and EUKF-C.
    a) shows the trace of the posterior covariance on a log scale with a zoomed-in inset showing the last 1000 steps of the simulation. 
    b) shows the relative error in the posterior covariance.
    Note that UKF overestimates the EnKF posterior covariance, whereas EUKF-A and EUKF-C posterior covariances closely track the EnKF posterior covariance as shown by b). 
    }
    \label{fig:MUKF_VDP_Pk}
\end{figure}

\begin{figure}[h]
    \centering
    \includegraphics[width = 1\columnwidth]{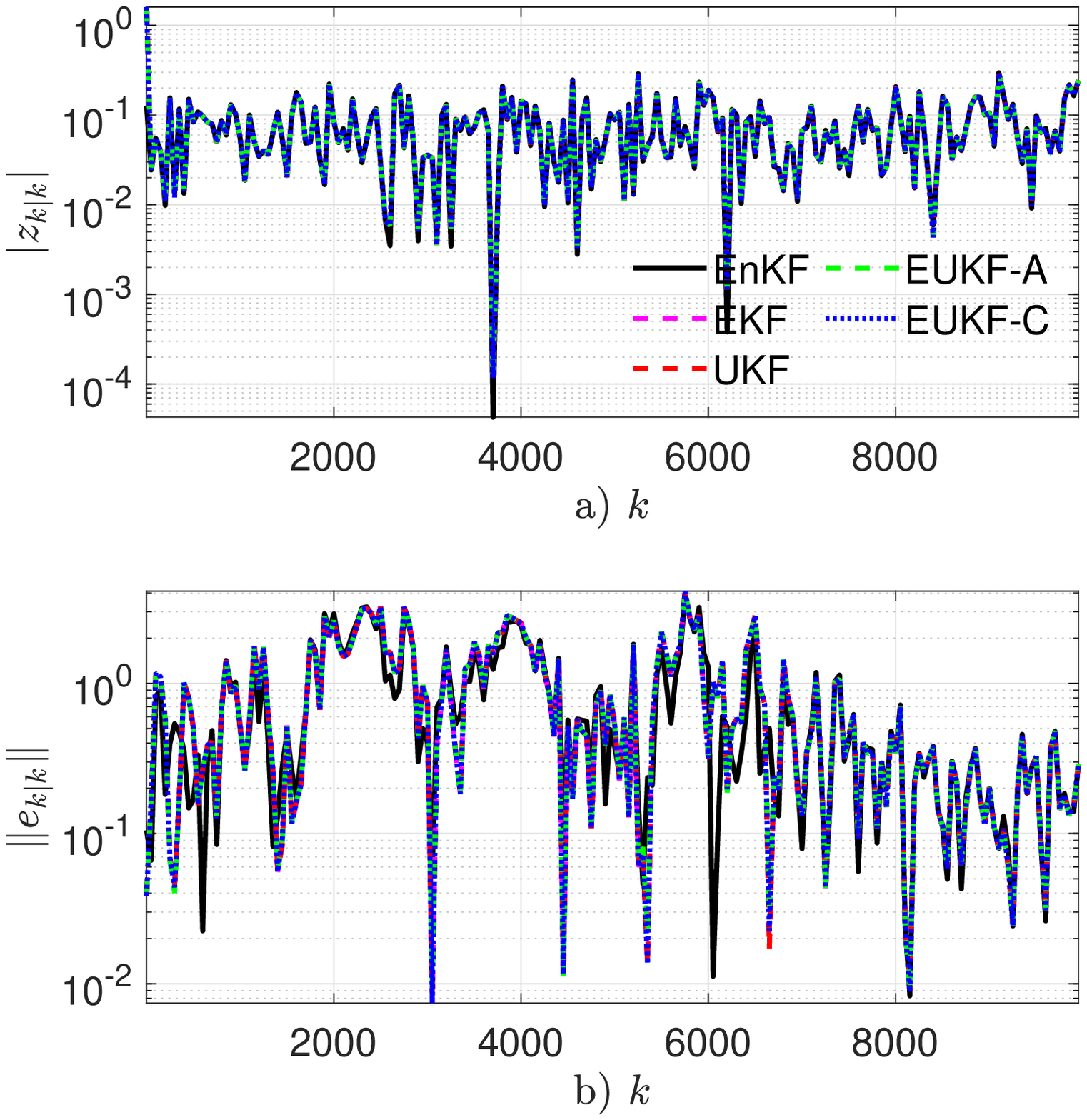}
    \caption{
    Example \ref{exmp:VDP}. Output error and posterior error computed using EnKF, EKF, UKF, EUKF-A, and EUKF-C.
    a) shows the output error and
    b) shows the norm of the posterior error. 
    Note that the output error and the posterior error computed by all algorithms are almost equal. 
    }
    \label{fig:MUKF_VDP_zk_ek}
\end{figure}

\begin{exmp}
    \label{exmp:Lorenz}
    \textit{Lorenz System}.
Consider the Lorenz system
\begin{align}
    \matl{c}
        \dot x_1 \\
        \dot x_2 \\
        \dot x_3
    \matr
        =
            \matl{c}
                \sigma (x_2-x_1) \\
                x_1(\rho-x_3)-x_2 \\
                x_1x_2 - \beta x_3
            \matr,
    \label{eq:Lorenz}
\end{align}
which exhibits a choatic behaviour for 
$\sigma = 10,$ $\rho = 28,$ and $\beta = 8/3$.
The Lorenz system \eqref{eq:Lorenz} is integrated using the forward Euler method with step size $T_\rms=0.01.$
Let the discrete system be modeled as
\begin{align}
    x_{k+1} 
        =
            f(x_k) + w_k,
\end{align}
where 
\begin{align}
    f(x) 
        \isdef 
            x+T_\rms 
            \matl{c}
                \sigma (x_2-x_1) \\
                x_1(\rho-x_3)-x_2 \\
                x_1x_2 - \beta x_3
            \matr
    \label{eq:Lorenz_disc}
\end{align}
and $w_k \sim \SN(0,Q_k).$
For all $k \ge 0$, let 
\begin{align}
    y_k 
        =
            C x_k + v_k,
    \label{eq:Lorenz_output}
\end{align} 
where
$C \isdef [0 \ 1 \ 0]$ and 
$v_{k} \sim \SN(0,R_k).$
For all $k\ge0,$ let $Q_k = 0.01 I_2$  and $R_k = 10^{-4}.$
Furthermore, let $x(0) = [1 \ 1 \ 1 ]^\rmT$ and $P_{0|0} = I_3$.
%

Letting $\alpha = 1.5$ in UKF, EUKF-A, and EUKF-C, 
Figure \ref{fig:MUKF_Lorenz_Pk}a) shows the trace of the posterior covariance computed by EnKF, EKF, UKF, EUKF-A, and EUKF-C. 
The true covariance is assumed to be given by EnKF with 100,000 ensemble members. 
Note that UKF overestimates the EnKF posterior covariance, whereas EUKF-A and EUKF-C closely track the EnKF posterior covariance and recover the EKF posterior covariance.  
Figure \ref{fig:MUKF_Lorenz_Pk}b)
shows the error of UKF, EKF, EUKF-A, and EUKF-C posterior covariance relative to EnKF. 
At the end of the simulation, UKF relative error is approximately $15\%$, whereas EUKF-A and EUKF-C relative error is less that $1\%.$

Figure \ref{fig:MUKF_Lorenz_zk_ek} shows the output error $z_{k|k}$ and the norm of the posterior error $e_{k|k}$ computed with all algorithms. 
Note that the output error and the posterior error are very close to each other. 
This example shows that the EUKF-A and EUKF-C posterior covariance estimate is more accurate than the UKF posterior covariance and is approximately equal to the EKF posterior covariance, however, the state estimates computed using all algorithms are almost equal. 
$\hfill\mbox{\Large$\diamond$}$
\end{exmp}

\begin{figure}[h]
    \centering
    \includegraphics[width = 1\columnwidth]{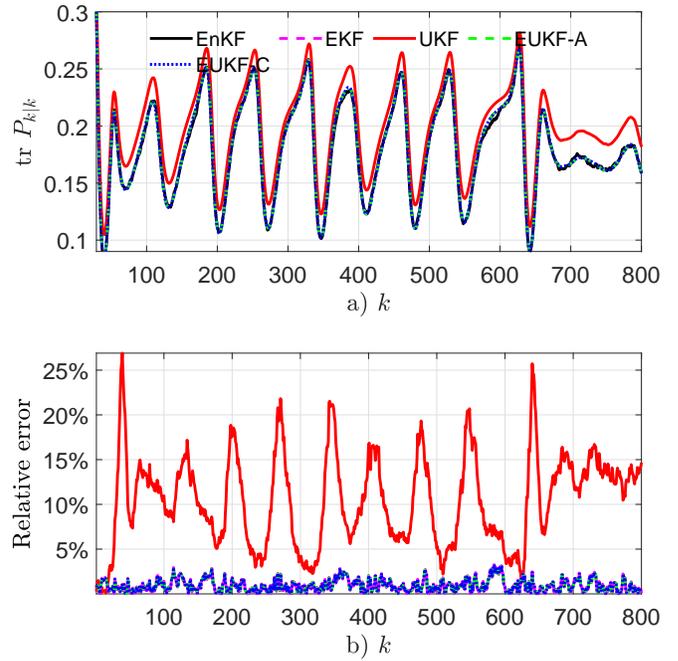}
    \caption{
    Example \ref{exmp:Lorenz}. 
    Posterior covariance computed using EnKF, EKF, UKF,  EUKF-A, and EUKF-C.
    a) shows the trace of the posterior covariance and
    b) shows the relative error in the posterior covariance.
    Note that UKF overestimates the EnKF posterior covariance, whereas EUKF-A and EUKF-C posterior covariances closely track the EnKF posterior covariance as shown by b). 
    }
    \label{fig:MUKF_Lorenz_Pk}
\end{figure}

\begin{figure}[h]
    \centering
    \includegraphics[width = 1\columnwidth]{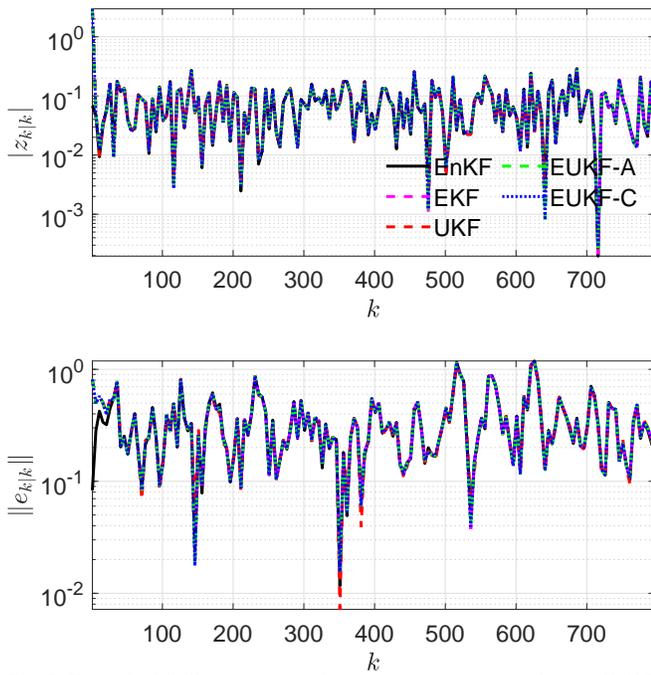}
    \caption{
    Example \ref{exmp:Lorenz}. Output error and posterior error computed using EnKF, EKF, UKF, EUKF-A, and EUKF-C.
    a) shows the output error and
    b) shows the norm of the posterior error. 
    Note that the output error and the posterior error computed by all algorithms are almost equal. 
    }
    \label{fig:MUKF_Lorenz_zk_ek}
\end{figure}

\section{CONCLUSIONS}
\label{sec:conclusion}

This paper presented two modifications of the UKF that specialize to the classical Kalman filter for linear systems. 
In linear systems, the two extensions are shown to be equivalent to Kalman filter.
In nonlinear systems, the two extensions provide more accurate estimate of the propagated posterior covariance in comparison to classical UKF as shown by the two numerical examples.
However, the accuracy of the state estimate is similar in all three filters. 



\renewcommand*{\bibfont}{\small}
\printbibliography

\end{document}